\documentclass[preprint,12pt]{imsart}

\RequirePackage[OT1]{fontenc}
\RequirePackage{amsthm,amsmath,subcaption,caption}
\RequirePackage[numbers]{natbib}
\RequirePackage[colorlinks,citecolor=blue,urlcolor=blue]{hyperref}
\usepackage{pythonhighlight}

\usepackage{tkz-tab}
\usepackage{amscd,amsfonts,amssymb,amsmath,latexsym,array,hhline,xcolor,graphicx, caption}
\usepackage{float}
\usepackage{appendix}
\usepackage{booktabs}
\usepackage{caption}
\usepackage{epstopdf}

\newcommand\cF{{\cal F}}

\newcommand\cM{{\cal M}}

\newcommand\R{{\bf R}}

\newcommand\Z{{\bf Z}}

\newcommand\N{\mathbb{N}}

\newtheorem{theo}{Theorem}[section]

\newtheorem{lemm}[theo]{Lemma}

\newtheorem{coro}[theo]{Corollary}

\newcommand\beq{\begin{equation}}
\newcommand\eeq{\end{equation}}
\newcommand\beqa{\begin{equation*}}
\newcommand\eeqa{\end{equation*}}
\newcommand\bea{\begin{eqnarray}}
\newcommand\eea{\end{eqnarray}}
\newcommand\bean{\begin{eqnarray*}}
\newcommand\eean{\end{eqnarray*}}

\DeclareMathOperator{\essinf}{ess\;inf}
\DeclareMathOperator{\esssup}{ess\;sup}

\DeclareMathOperator{\supp}{supp}

\begin{document}

\begin{frontmatter}

\title{A short note on super-hedging an arbitrary number of European options with integer-valued strategies }

\author[A1]{ Dorsaf Cherif,}
\author[A2]{ Meriam EL MANSOUR, }

\author[A3]{ Emmanuel LEPINETTE}

\address[A1]{ Faculty of Sciences of Tunis, Tunisia.\\
Email: dorsaf-cherif@hotmail.fr
}

\address[A2]{ Ceremade, UMR  CNRS 7534,  Paris Dauphine University, PSL National Research, Place du Mar\'echal De Lattre De Tassigny, \\
75775 Paris cedex 16, France and\\
 Faculty of Sciences of Tunis, Tunisia.\\
Email: elmansour@ceremade.dauphine.fr
}
\address[A3]{ Ceremade, UMR  CNRS 7534,  Paris Dauphine University, PSL National Research, Place du Mar\'echal De Lattre De Tassigny, \\
75775 Paris cedex 16, France and\\
Gosaef, Faculty of Sciences of Tunis, Tunisia.\\
Email: emmanuel.lepinette@ceremade.dauphine.fr\\ \bigskip
%{\color{blue}Emmanuel Lepinette is Associate Professor at Paris-Dauphine University where he works in mathematical finance, in particular stochastic financial market models with transaction costs and  pricing method without martingale measures. See \url{https://sites.google.com/view/emmanuel-lepinette/research-cv-and-others}\\ \smallskip
%Dorsaf Cherif is a PhD student under the supervision of E. Lepinette and Meriem El Mansour, a former PhD student of E. Lepinette, is lecturer at Tunis-Dauphine Institute, Tunisia.} \bigskip
}

\begin{abstract} The usual theory of asset pricing in finance assumes that the financial strategies, i.e. the quantity of risky assets to invest, are real-valued so that they are  not integer-valued in general, see the Black and Scholes model for instance. This is clearly contrary to what it is possible to do in the real world. Surprisingly, it seems that there is no many contributions in that direction in the literature, except for a finite number of states. In this paper, for arbitrary $\Omega$, we show that, in discrete-time, it is possible to evaluate the minimal super-hedging price  when we restrict ourselves to integer-valued strategies. To do so, we only consider terminal claims that are continuous piecewise  affine functions of the underlying asset. We formulate a dynamic programming principle that can be directly implemented on an historical data and which also provides  the optimal integer-valued strategy. The problem with general payoffs remains open but should be solved with the same approach.
 \end{abstract}

\begin{keyword} Super-hedging prices--European options--Integer-valued strategies--Optimal super-hedging prices--Dynamic programming principle.

\end{keyword}

\noindent {\bf MSC:}{ \rm 60-08, 60H30, 90-05, 91-05, 91-10, 91-G15}

\end{frontmatter}

\section{Introduction}

The problem of super-hedging a European claim, such as a call option, is very classical in mathematical finance but it has only been  solved for real-valued strategies so that the optimal strategy, corresponding to the minimal hedging or super-hedging price, is not integer-valued contrary to what it is allowed to do in the real life. This is why we propose  to solve the problem  of super-hedging a European claim with  integer-valued financial strategies.

Actually, the main contribution in the literature for integer-valued financial strategies is the paper \cite{GK} where a finite set of states $\Omega$ is considered. In this setting, the authors adapt the usual theory for the real-valued strategies to the integer-valued ones, i.e. they introduce a no-arbitrage condition which is equivalent to the existence of a risk-neutral probability measure and show that the minimal super-hedging price may be characterized through the martingale measures similarly to the usual case with real-valued strategies. See also \cite{DLW} where $\Omega$ is finite.  The general case with an arbitrary state space $\Omega$ is still an open problem. Also, portfolio optimization problems of Markowitz type are considered in \cite{BI}, \cite{BL}, \cite{BT}.

Let us recall that the usual approach of pricing assumes that the financial market model satisfies a no-arbitrage condition NA, which is equivalent to the existence of a risk-neutral probability measure $Q$ under which the discounted asset prices are martingales, see the Dalang-Morton-Willinger theorem \cite{DMW}. Under NA, we may show that there exists a minimal super-hedging price $P_0^*(\xi_T)$ for the European claim $\xi_T\ge 0$ given by 
\bea P_0^*(\xi_T)=\sup_{Q\in \cM(P)}E_Q(\xi_T),\label{MP}\eea
where $\cM(P)$ is the set of all risk-neutral probability measures equivalent to the initial probability measure $P$ of the model. Here, we suppose that the risk-free interest rate  of the model is $r=0$. Recall that the formula above holds in discrete time but also in continuous time with extra-conditions on the model. Indeed,  the no-arbitrage condition needs to be strenghtened and it   is only equivalent to the existence of $Q\sim P$ under which discounted asset prices are local martingales, see \cite{DS1}, \cite{DS2}, \cite{DS3}.

In any case, the optimal strategy that achieves the minimal super-hedging price (\ref{MP}) is not, in general, integer-valued. The typical example is the continuous-time Black and Scholes model where the so-called delta-hedging strategy for the European call is explicit and lies in the set $[0,1]\setminus \{0,1\}$, see \cite{BS}.

Clearly, a new approach is necessary to compute the super-hedging prices for only integer-valued financial strategies. We follow the ideas developed in \cite{CL} where the problem is initially solved without any no-arbitrage conditions. Then, a no-arbitrage condition AIP naturally appears and means that the infimum price of the zero claim (non negative claims  more generally) is not $-\infty$. This condition is clearly necessary for numerical computations. Actually,  it is shown that AIP is equivalent to the property that the infimum super-hedging price of any non negative claim is non negative, as observed in the real markets. In our paper, we do not explicitly suppose such a no-arbitrage condition but the form of the conditional supports of the asset price  we assume implies this condition for the model with real-valued strategies. 

Our paper is a first attempt to compute super-hedging prices with only integer-valued financial strategies. We restrict ourselves to payoffs which are piecewise affine functions of the underlying asset and we assume specific conditional supports for the asset prices. The problem for general payoff functions remains open but we conjecture that our proposed method can be adapted with more technicalities. Also, problems such as characterizations of the no-arbitrage condition NA with only integer-valued strategies or generalization of our work to arbitrary conditional supports of the asset prices remains open if $\Omega$ is not finite. \bigskip

\noindent {\bf Notations} If $A\subset \R$ is a random subset of $\R$ and $\cF$ is a $\sigma$-algebra, we denote by $\mathbb{L}^0(A,\cF)$ the family of all $\cF$-measurable random variables $X$ on the probability space $(\Omega,\cF,P)$ such that $X\in A$ a.s..

\section{The super-hedging problem:}
We consider $n\geq 1$ options that we want to super-replicate in discrete time $t=\lbrace 0,...T \rbrace .$
Let $(\Omega,( \mathcal{F}_t)_{0\leq t\leq T})$ be a stochastic basis where $( \mathcal{F}_t)_{0\leq t\leq T}$ is supposed to be complete.
We consider a financial market model composed of two assets. We suppose, without loss of generality, that the risk-free asset  is $S_t^0=1$ for all $t \in \lbrace 0,...T \rbrace,$ while the risky asset price is described by a stochastic process $S=(S_t)_{0\leq t\leq T}.$ Recall that a self-financing portfolio process $(V_t)_{0\leq t\leq T}$ satisfies by definition:
$$ \Delta V_t= V_t - V_{t-1}= \theta_{t-1} \Delta S_t,\quad  t=1,\cdots,T,$$ where  $\theta_{t-1}$ is $\cF_{t-1}$-measurable and represents the number of risky assets of the portfolio.

In this paper, we consider European options whose payoffs are of the form $\xi_T=g(S_T) \in \mathbb{L}^0(\R,\mathcal{F}_T),$ where $g$ is a continuous piecewise affine function. The typical example is the European  call  option, i.e. $g(x)=(x-K)^+$, $K>0$.
Our goal is to compute the set of all super-hedging prices of $\xi_T$, i.e the set of all $V_0$,  initial values of  self-financing portfolio processes $(V_t)_{0 \leq t \leq T}$, such that $V_T \geq \xi_T$ almost surely.
Contrarily to what it is usual to do in the literature, we restrict ourselves to the case of integer-valued strategies, i.e $\theta_t \in \Z$ almost surely, for all $ t \in \lbrace 0,...T \rbrace,$ where $\Z=\N\cup(-\N)$ and $\N$ is the set of all non negative integers.  In the case of super-hedging an arbitrary number of options $n \geq 1,$ the problem reads as $ V_T \geq n \xi_T, \text{~a.s.}$ and it is clearly interesting to analyse the impact of $n$ on the strategies and the infimum prices, as linearity is not necessarily preserved with respect to the quantity $n$ of claims. 

To solve this problem, we follow the approach of \cite{CL} and \cite{BCL} that we adapt to integer-valued strategies. To do so, we first solve backwardly the super-hedging problem between two dates $t-1$ and $t,$ and we show that the procedure may be propagated backwardly as the minimal super-hedging price we obtain at time $t-1$ is still a continuous piecewise affine payoff function of the underlying asset.
It  is then possible to deduce the minimal super-hedging price at time $t=0$ by induction.

\section{The  super-hedging problem.}\label{section3}
Let $ t \leq T $ and $g_t$ be a continuous piecewise affine function, i.e. there exists a subdivision  $0=a_0 < a_1<...<a_{n-1}<a_n = \infty$ of $[0,\infty]$ such that $g_t$ is an affine function on the interval $ [a_i,a_{i+1}), \forall i \in \lbrace 0,...,n-1 \rbrace.$ As the asset prices are non negative, we suppose without loss of generality that $a_0=0.$
We first solve the  one step problem: find $V_{t-1}$ and the strategy $\theta_{t-1}\in \mathbb{L}^0(\R,\cF_{t-1})$ such that:

$$V_{t-1} + \theta_{t-1} \Delta S_t \geq  g_t(S_t),\,{\rm a.s..}$$
This is equivalent to: 

\bean  V_{t-1} &\geq&  g_t(S_t)-\theta_{t-1} \Delta S_t, \\
\Leftrightarrow  V_{t-1} &\geq & g_t(S_t)-\theta_{t-1} S_t+ \theta_{t-1} S_{t-1},\\
\Leftrightarrow  V_{t-1} &\geq &  \esssup_{\mathcal{F}_{t-1}}(  g_t(S_t)-\theta_{t-1} S_t )  + \theta_{t-1} S_{t-1}.
\eean

Recall that the condition essentiel supremum $\esssup_{\mathcal{F}_{t-1}}(\Gamma)$, for a familly of random variables $\Gamma$,  is the smallest random variable that dominates a.s. any $\gamma\in \Gamma$, see the definition in \cite{CL}. Equivalently, we may show that:
\begin{equation}
 V_{t-1} \geq V_{t-1}(\theta_{t-1}) := \underset{x \in \supp_{\mathcal{F}_{t-1}}(S_t)} \sup(  g_t(x)-\theta_{t-1} x )  + \theta_{t-1} S_{t-1}, 
\end{equation}
where $\supp_{\mathcal{F}_{t-1}}(S_t)$ is the conditional support of $S_t$ knowing $\mathcal{F}_{t-1}$, see \cite{CL} and \cite{EL} for the definition and the proof of the inequality above. \smallskip

In the following, we suppose that there exist two deterministic numbers $k_{t-1}^d \in (0,1)$ and $k_{t-1}^u \in (1, \infty)$ such that $\supp_{\mathcal{F}_{t-1}}(S_t)=[k_{t-1}^d S_{t-1}, k_{t-1}^u S_{t-1}].$ This model may be seen as a generalization of the Binomial model and the conditions imposed on the coefficients $k_{t-1}^d$ and $k_{t-1}^u$ are equivalent to the weak no-arbitrage condition AIP, see \cite{CL}. In particular, we have:
$$V_{t-1}(\theta_{t-1})= \underset{x \in [k_{t-1}^d S_{t-1}, k_{t-1}^u S_{t-1}]} \sup(  g_t(x)-\theta_{t-1} x ) + \theta_{t-1} S_{t-1}\in \mathbb{L}^0(\R,\cF_{t-1}). $$  
We define $V_{t-1}^*$ as the infimum of all  super-hedging prices at time $t-1$ over all integer-valued strategies in $\Z$, i.e.

$$V_{t-1}^*:= \underset{\theta_{t-1} \in \mathbb{L}^0(\Z,\mathcal{F}_{t-1})} \essinf V_{t-1}(\theta_{t-1}).$$

\begin{lemm}
We have 

 \begin{equation}V_{t-1}^*=\underset{\theta \in \Z} \inf V_{t-1}(\theta).\label{2}
\end{equation}
\end{lemm}
\begin{proof}

 Let us define $\gamma=\underset{\theta \in \Z} \inf V_{t-1}(\theta)\in L^0(\R,\cF_{t-1})$, see \cite{CL}. As $ V_{t-1}^* \leq V_{t-1}(\theta),$ for all $\theta \in \Z$, we get that  $V_{t-1}^*\leq \underset{\theta \in \Z} \inf  V_{t-1}(\theta)=\gamma.$ On the other hand, if $\theta_{t-1} \in \mathbb{L}^0(\Z, \mathcal{F}_{t-1}),$ then:
\bean \theta_{t-1}&=&\underset{ \theta \in Z} \sum \theta 1_{\lbrace \theta_{t-1} = \theta  \rbrace },\\
 V_{t-1}(\theta_{t-1})&=&\underset{ \theta \in Z} \sum  V_{t-1}(\theta) 1_{\lbrace \theta_{t-1} = \theta  \rbrace }\ge \underset{ \theta \in Z} \sum  \gamma 1_{\lbrace \theta_{t-1} = \theta  \rbrace }=\gamma.\eean
We deduce that $V_{t-1}^*\ge \gamma$ and the conclusion follows.
\end{proof}

In the following, we first solve the super-hedging problem between two dates.
\begin{theo}[One step problem]\label{theoreme principal}
Let us consider $t \in \lbrace 	1,...,T \rbrace$ and   suppose that $\xi_t=g_t(S_t)$ where $g_t$ is a continuous piecewise affine function. Moreover, we assume   that there exists  two deterministic numbers $k_{t-1}^d \in (0,1)$ and $k_{t-1}^u \in (1, \infty)$ such that $$\supp_{\mathcal{F}_{t-1}}(S_t)=[k_{t-1}^d S_{t-1}, k_{t-1}^u S_{t-1}].$$
Then, $V^*_{t-1}=g_{t-1}(S_{t-1})$ where $g_{t-1}$  is a continuous piecewise linear function. 
\end{theo}

\begin{proof}
By assumption, there exist a subdivision $(a_i)_{i=0,\cdots,n}$ of $[0, \infty]$, with $a_0=0<a_1<...<a_{n-1}<a_n=\infty$, such that $g_{t}$ is an affine function on each interval. Let us define  
$$x_i(S_{t-1}) =(k_{t-1}^d S_{t-1} \vee a_i)\wedge k_{t-1}^u S_{t-1},\quad  i = 0,\cdots,n. $$
It is straightforward that
$$ V_{t-1}(\theta_{t-1})=\underset{i = 0,\cdots,n } \sup \left[ g_t(x_i(S_{t-1}))-\theta_{t-1}x_i(S_{t-1}) \right] + \theta_{t-1}S_{t-1}.$$
Note  that $x_0(S_{t-1})=k^d S_{t-1}$ and $x_n(S_{t-1})=k^u S_{t-1}$  and some terms of the sequence $(x_i)_i$ may coincide. Let us define the functions
$$ h^i(\theta_{t-1},S_{t-1})=g_t(x_i(S_{t-1}))
+\theta_{t-1}(S_{t-1}-x_i(S_{t-1})),\, i=0,\cdots,n.$$

The slopes of the affine functions $ \theta_{t-1} \mapsto h^i(\theta_{t-1},S_{t-1})$
are  given by the non decreasing sequence $(S_{t-1}-x_i(S_{t-1}))_{i = n,n-1,\cdots,0}$ such that $S_{t-1}-x_n(S_{t-1}) < 0$ and $S_{t-1}-x_0(S_{t-1}) > 0.$

By ordering the indices in the decreasing order, we obtain $(n+1)$ affine functions $\theta_{t-1} \mapsto h^i(\theta_{t-1})$ for $i \in \lbrace n,n-1,...,1,0 \rbrace$ with increasing slopes $(S_{t-1}-x_i)_{i \in \lbrace n,n-1,...,1,0 \rbrace},$ such that: $S_{t-1}-x_n < 0$ and $S_{t-1}-x_0  > 0.$ Therefore, the mapping  $V_{t-1}:\theta_{t-1} \mapsto \underset {i = n,...,0 } \sup h^i(\theta_{t-1}, S_{t-1})$ is a piecewise affine function, i.e.  there  exists a subdivision:
 $$ -\infty=\alpha_0 <\alpha_1(S_{t-1})  \leq ...\leq \alpha_{m-1}(S_{t-1}) <\alpha_m=\infty  $$
  such that $V_{t-1}$ is an affine function of $\theta_{t-1}$ on each interval $[\alpha_i(S_{t-1}), \alpha_{i+1}(S_{t-1})]$,\quad $i=0,\cdots,m-1$.  Note  that the function $V_{t-1}$ is convex in $\theta_{t-1}$ and the elements of the partition define the intersection points between two distinct and successive graphs of the affine functions $h^{i+1}(\theta_{t-1}),h^{i}(\theta_{t-1})$.  So, there exists $\theta_{t-1}^* \in [ \alpha_1(S_{t-1}) -1,...,\alpha_{m-1}(S_{t-1})+1 ] \cap \Z$ such that:
$$\underset{\theta_{t-1} \in \Z} \inf V_{t-1}(\theta_{t-1})=V_{t-1}(\theta_{t-1}^*).$$
It remains to evaluate $\alpha_1(S_{t-1})$ and $\alpha_{m-1}(S_{t-1})$. To do so, let us  solve the  equations $h^i(\alpha)=h^j(\alpha)$, $i,j = 0,...,m$ and  $x_i \neq x_j$. Since we suppose that  $x_i-x_j \neq 0,$ we get that 
$$\alpha= \frac{g_t(x_i)-g_t(x_j)}{x_i-x_j}. $$ We deduce that $|\alpha|\leq L_{t}$ where  $L_{t} >0$ is  a Lipschitz constant of the piecewise affine function $g_t$. We deduce that  $\theta_{t-1}^* \in [-L_{t}-1,L_t+1] \cap \Z$ and $$g_{t-1}(S_{t-1})=V_{t-1}(\theta_{t-1}^*(S_{t-1}))=\underset{\theta_{t-1} \in [-L_{t}-1,L_{t}+1] \cap \Z} \min~ \underset{i =0,\cdots,n } \sup h^i(\theta_{t-1},S_{t-1}).$$
We conclude that $g_{t-1}$ is a continuous piecewise affine function as a  minimum of  a finite number of continuous piecewise affine functions. 
\end{proof}

\begin{coro}(The multi-period super-hedging problem)\label{coroMPS}
Suppose that, at time $T>0$, the payoff is $\xi_T=g_T(S_T)$ where $g_T$ is a continuous piecewise  affine function. Moreover, we assume   that there exists   deterministic numbers $k_{t-1}^d \in (0,1)$ and $k_{t-1}^u \in (1, \infty)$ for each $t=1,\cdots,T$ such that we have $\supp_{\mathcal{F}_{t-1}}(S_t)=[k_{t-1}^d S_{t-1}, k_{t-1}^u S_{t-1}].$ Then,  there exists a minimal super-hedging portfolio process $(V^*_{t})_{t=0,\cdots,T}$ such that $V_T^*\ge \xi_T$. We have  $V^*_{t}=g(t,S_{t})$ where $g(t,\cdot)$  is a continuous piecewise affine function given by 
\bean g(t,s)&=&\underset{\theta \in [-L_{t+1}-1,L_{t+1}+1] \cap \Z} \min~ \underset{i =0,\cdots,n^{(t+1)} } \sup \left( g(t+1,x_i^{(t+1)}(s))
+\theta (s-x_i^{(t+1)}(s)) \right),\\
x_i^{(t+1)}(s) &=&(k_{t}^d s \vee a_i^{(t+1)})\wedge k_{t}^u s,\quad  i = 0,\cdots,n^{(t+1)},
\eean 
where $L_{t+1}$ is any Lipschitz constant of the function $g(t+1,\cdot)$ and $(a_i^{(t+1)})_{i=0,\cdots,n^{(t+1)}}$ is any partition such that $a_0^{(t+1)}=0<a_1^{(t+1)}<...<a_{n^{(t+1)}-1}^{(t+1)}<a_{n^{(t+1)}}^{(t+1)}=\infty$ and $g(t+1,\cdot)$ is an affine function on  $[a_i^{(t+1)}, a_{i+1}^{(t+1)} ]$, $i\le n^{(t+1)}-1$. The associated super-hedging strategy $\theta^*$  is given by the argmin of the minimisation problem defining $g(t,\cdot)$ in the expression above.

\end{coro}

The result above provides a recipe to compute backwardly the minimal super-hedging price at any time, which is the main goal of our paper.

\begin{theo} Let $g(t,x,n)$ be the price function as given in Corollary \ref{coroMPS} at time $t$ for $n\ge 1$ units of claims $g_T(S_T)$, i.e. the payoff is $\xi_T=ng_T(S_T)$, where $g_T$ satisfies the conditions of Corollary \ref{coroMPS}. Let $\hat g(t,x)$ be the price function  for one unit of the claim  $g_T(S_T)$   in the model where real-valued strategies are allowed, see \cite{CL}. Suppose that the price process $S$ satisfies the conditions of Corollary \ref{coroMPS}. Then, $\frac{g(t,x,n)}{n}$ converges to $\hat g(t,x)$ as $n\to \infty$.
\end{theo}
\begin{proof} Note that there exists a constant $M>0$ such that $0\le S_t \le M$ a.s. for all $t=0,\cdots,T$.  By definition, there exists an integer-valued strategy $(\theta^n_t)_{t=0,\cdots,T-1}$ such that $g(t,x,n)+\sum_{t=1}^T\theta^n_{t-1}\Delta S_t\ge ng_T(S_T)$. By definition of the smallest price for $n$ units of claims  with real-valued strategies, we have   $\hat g(t,x,n)=n\hat g(t,x)$ where $\hat g(t,x)=\hat g(t,x,1)$ by linearity of the model \cite{CL}. As it is the smallest price, we deduce that   $n\hat g(t,x)\le g(t,x,n)$, i.e. $\hat g(t,x)\le \frac{g(t,x,n)}{n}$. On the other hand, still by definition, there exists a real-valued strategy $(\theta^r_t)_{t=0,\cdots,T-1}$ such that $n\hat g(t,x)+\sum_{t=1}^T\theta^r_{t-1}\Delta S_t\ge ng_T(S_T)$. We decompose each $\theta^r_{t-1}=\bar\theta^r_{t-1}+d_{t-1}$ into $\bar\theta^r_{t-1}$  the integer part of $\theta^r_{t-1}$ and $d_{t-1}=\theta^r_{t-1}-\bar\theta^r_{t-1}$ the residual term. We observe that $\left |\sum_{t=1}^Td_{t-1}\Delta S_t \right|\le 2TM$. We deduce that $2TM+n\hat g(t,x)+\sum_{t=1}^T\bar\theta^r_{t-1}\Delta S_t\ge ng_T(S_T)$ and we deduce by definition that $g(t,x,n)\le 2TM+n\hat g(t,x)$. We finally get the inequality  $\hat g(t,x)\le \frac{g(t,x,n) }{n}\le \frac{2TM}{n}+\hat g(t,x)$. The conclusion follows.

\end{proof}

The theorem above is interesting for practitioners. It implies that, if $n$ is large enough and if the practitioners only trade integer-valued strategies, then the price of $n$ claims $g_T(S_T)$ is approximately $n\hat g(t,x)$. This is important for numerical reasons as the computation of $g(t,x,n)$ is rather time-consuming contrarily to the computation of $\hat g(t,x)$, see \cite{CL}.

\subsection{Example in the one step problem: the case of the Call option}\label{Ex1}

At time $t=T$, suppose that  the payoff is  $\xi_T^n=ng(S_T)$ where $n\ge 1$ and $g(x)=(x-K)^+$,  $K=500$. We suppose that $\supp_{\mathcal{F}_{T-1}}(S_T)=[k^d S_{T-1}, k^u S_{T-1}]$ for some constants $k^d, k^u$ such that $0<k^d<1<k^u$. Precisely, we suppose that $k^d=0.9$ and $k^u=1.2$. Observe that the super-hedging problem $V_{T-1}+ \theta_{T-1} \Delta S_T \ge n g(S_T)$  is equivalent to  
 \bean V_{T-1} \ge V_{T-1}(\theta_{T-1}) =  \max_{  k \in \{k^d , k^u \} }[n g(k S_{T-1})-\theta_{T-1} k S_{T
 -1}] + \theta_{T-1} S_{T-1}.
 \eean

In the following we give the explicit expression of $V_{T-1}(\theta_{T-1})=V_{T-1}^n(\theta_{T-1})$.
 
 \bigskip
 
   If $k^u \le K/S_{T-1}$, i.e. $S_{T-1}\le K/k^u$, then 
  
  \[V_{T-1}(\theta_{T-1}) = \begin{cases}
\theta_{T-1} S_{T-1}(1-k^u) & \text{if}~~ \theta_{T-1} \le 0, \\
\theta_{T-1} S_{T-1}(1-k^d) & \text{if}~~ \theta_{T-1} \ge 0.
\end{cases}\]

 Therefore,   $\theta_{T-1}^{*,n}(S_{T-1}) = 0$,   and $V_{T-1}^{*,n}(S_{T-1})=V_{T-1}(\theta_{T-1}^{*,n})=0$.

   \bigskip
   
If $k^d \ge K/S_{T-1}$, i.e. $S_{T-1}\ge K/k^d$,

 \[V_{T-1}(\theta_{T-1}) = \begin{cases}
\theta_{T-1} S_{T-1}(1-k^u)+n k^u S_{T-1}- n K & \text{if}~~ \theta_{T-1} \le n, \\
\theta_{T-1} S_{T-1}(1-k^d)+n k^d S_{T-1}- n K & \text{if}~~ \theta_{T-1} \ge n.
\end{cases}\]

 We conclude that $\theta_{T-1}^{*,n}(S_{T-1}) = n$,   and $V_{T-1}^{*,n}(S_{T-1}) =n(S_{T-1}-K)$.
 
 \bigskip
 
If $k^d \le K/S_{T-1} \le k^u$, i.e. $S_{T-1}\in [K/k^u,K/k^d]$,

 \[V_{T-1}(\theta_{T-1}) = \begin{cases}
\theta_{T-1} S_{T-1}(1-k^u)+n k^u S_{T-1}- n K & \text{if}~~ \theta_{T-1} \le \frac{n k^u S_{T-1} -n K }{S_{T-1}(k^u-k^d)}, \\
\theta_{T-1} S_{T-1}(1-k^d) & \text{if} ~~ \theta_{T-1} \ge \frac{n k^u S_{T-1} -n K }{S_{T-1}(k^u-k^d)}.
\end{cases}\]

 Let us define  
 \bean \alpha_{T-1}^n(S_{T-1})&:=&\frac{n k^u S_{T-1} -n K }{S_{T-1}(k^u-k^d)},\\
 f^n(x,S_{T-1})&:=& xS_{T-1}(1-k^d)1_{\{x \geq \alpha_{T-1}(S_{T-1})\}}\\
 &&+(xS_{T-1}(1-k^u)+nk^uS_{T-1}-nK )1_{\{x < \alpha_{T-1}(S_{T-1})\}}.
 \eean
 We denote by $\lfloor \alpha_{T-1}^n(S_{T-1}) \rfloor$ the lower integer part of $ \alpha_{T-1}^n$. Then, 
 \[\theta_{T-1}^{*,n}(S_{T-1}) = \begin{cases}
\lfloor \alpha_{T-1}^n(S_{T-1}) \rfloor  & \text{if} ~~f^n(\lfloor \alpha_{T-1}(S_{T-1}) \rfloor) \le f^n(\lfloor \alpha_{T-1}(S_{T-1}) \rfloor +1), \\
\lfloor \alpha_{T-1}^n(S_{T-1}) \rfloor+1 & \text{otherwise.} 
\end{cases}\] 

 So, we have:
  \bean V_{T-1}^{*,n}(S_{T-1})&=&(\lfloor \alpha_{T-1}^n(S_{T-1}) \rfloor  S_{T-1}(1-k^u)+nk^uS_{T-1}-nK )1_{G_{T-1}^n}(S_{T-1})\\
  &&+(\lfloor \alpha_{T-1}^n \rfloor +1) S_{T-1}(1-k^d)1_{(G_{T-1}^n)^c}(S_{T-1}),
 \eean
where \bean G_{T-1}^n&:=& \{S:~f^n(\lfloor \alpha_{T-1}^n(S) \rfloor) \le f^n(\lfloor \alpha_{T-1}^n(S) \rfloor +1)\}=\{ \lfloor \alpha_{T-1}^n (S)\rfloor \le \beta_{T-1}^n(S)\},\\
\beta_{T-1}^n(S)&:=&\alpha_{T-1}^n(S)+\frac{1-k^d}{k^d-k^u}.
\eean

 A graphic illustration of $V_{T-1}^{*,n}/n$ as a function of $S_{T-1}$ is given in Figure \ref{F}.
 
  \begin{figure}[!h]
\begin{center}
\includegraphics[scale=0.4]{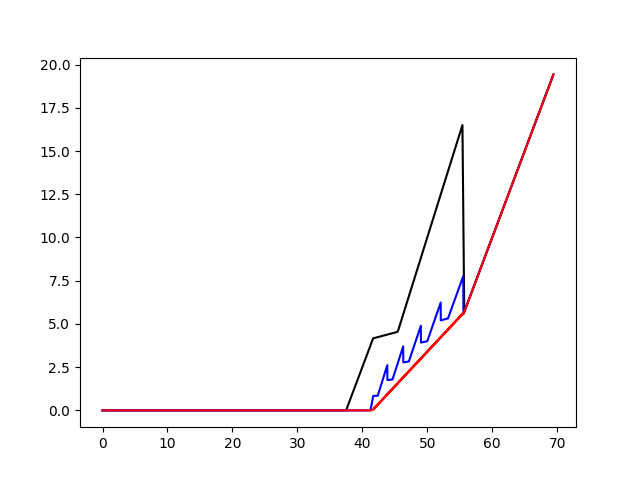} 
\caption{The function $x=S_{T-1}\mapsto g(T-1,x,n)/n=V_{T-1}^{*,n}/n$, $K=500$, for $n=1$ (black), $n=5$ (blue), $n=100$ (red).}
\label{F}
%\caption{pp}
 \end{center} 
 \end{figure}
 
 We observe that $V_{T-1}^{*,n}$ is not a convex function of $S_{T-1}$ even if the payoff function is and, moreover, $g(T-1,x,n)\ne n g(T-1,x,1)$.

 \newpage

\section{Numerical illustration}

In this section, we illustrate the method developed above when the underlying asset $S$ is the french CAC 40 index and the European claim is the Call option. The historical data is composed of daily observations of the CAC 40 values between the  6th of June 2019    and the 16th of June 2021. We use the two first years of the data set  to calibrate the model while we implement the model on the third year. Here, we suppose that $\supp_{\cF_t}S_{t+1}=[k^d_tS_t, k^u_tS_t]$ where $k^d_t$ and $k^u_t$ are estimated as follows:

\bean k^d_t&=&\min_{i=j,\cdots,N} S_{t+1}^{(j)}/S_t^{(j)},\\
k^u_t&=&\max_{i=j,\cdots,N} S_{t+1}^{(j)}/S_t^{(j)},
\eean
where $N$ is the number of training periods and $S_t^{(j)}$ are the observed values at time $t$ during the j-th periods. The algorithms are written in Python. The main difficulty is to write a code whose execution time is reasonable. Indeed, recall  that the price function $g(t,x)$ is computed backwardly from $g(t+1,x)$. If this function $g(t,x)$  is naively  coded from $g(t+1,x)$, then the computation  may take more than two weeks ! So it is better to approximate, at each step, the function $g(t,x)$ as a numpy array consisting of discretized values following a grid $(x_i)_{i=0}^{N_t}$ where $x_i=step* i$. Here, we choose $step=0.1$ and $N_t$ is chosen so that $x_{N_0}\le S_0^{{\rm max}}$ where $S_0^{{\rm max}}$ is the maximal value for $S_0$ that we observe in our data. At last, $x_{N_t}\le S_0^{{\rm max}}*(\max_{r\le t}k^u_r)^t$. 

The relative hedging error is defined as $\epsilon_T= 100*(V_T-g(T,S_T))/S_T$ where $(V_t^*)_{t=0,\cdots,T}$ is the optimal super-hedging portfolio process whose initial value is the minimal super-hedging price, as computed in the last section. We present in Figure \ref{F1}.1  the distribution of $\epsilon_T$ when $n=1$.  Of course, we expect that  $\epsilon_T\ge 0$ a.s. and this is confirmed on our test data set. Note that, we could have observed some negative values as the model is  calibrated from data values anterior to the test data set.

 \begin{figure}[!h]
\centering
\begin{minipage}{0.4\textwidth}
%\begin{subfigure}[t]{0.4\textwidth}
\includegraphics[width=\textwidth]{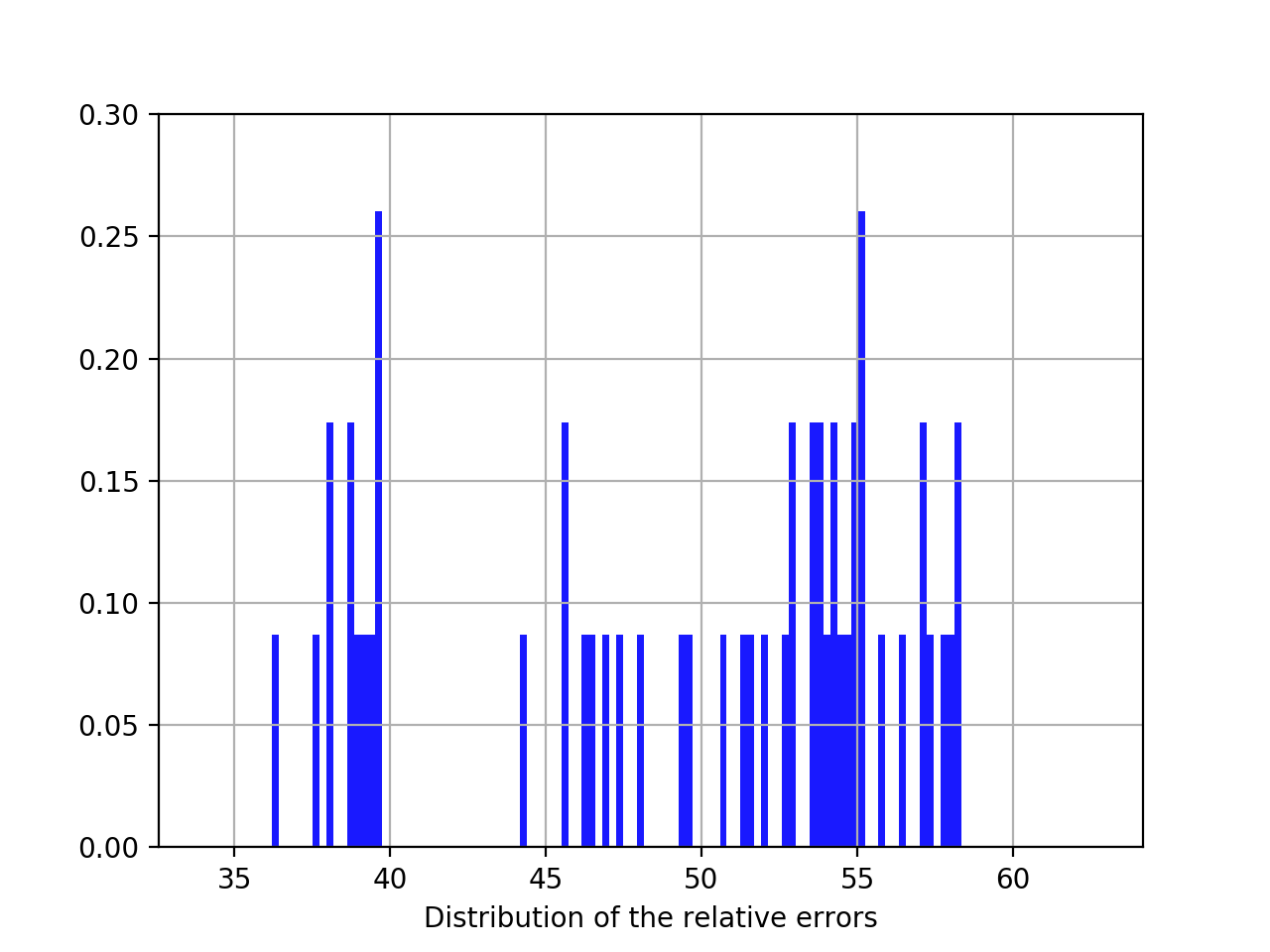}

%\caption{{\tiny \\
%Distribution of the super-hedging error:  $n=1$ and $K=3000$.}}
\label{F1}
%\end{subfigure}
\end{minipage}
\begin{minipage}{0.4\textwidth}
%\begin{subfigure}[t]{0.4\textwidth}
\includegraphics[width=\textwidth]{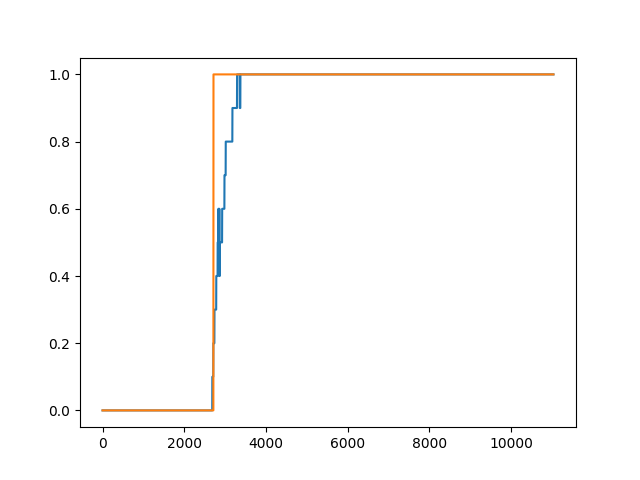}
%\label{F2}
%\caption{{\tiny Comparison of the optimal strategies per unit of claims for $n=1$ and $n=10$.}}
\end{minipage}
\caption{F2.1: Super-hedging errors with  $n=1$ and $K=3000$. F2.2: Comparison of the optimal strategies per unit of claims for $n=1$ and $n=10$. }
\label{F1}
\end{figure}

 Let us denote by $g(t,x,n)$ the price function at time $t$ of  the optimal portfolio  process, i.e. $V_t^{*,n}=g(t,S_t,n)$ such that $V_T^{*,n}\ge \xi_T^n$ a.s.,  when the European claim is $\xi_T^n:=n*(S_T-K)^+$. The natural question is the following: Do we have $g(t,x,n)=ng(t,x,1)$ ? The answer is yes when real-valued strategies are allowed since the  hedging problem is then linear with respect to the number of claims.

In the case of integer-valued strategies, the answer is not trivial. It  is actually negative, see the first example above. By definition of the  infimum super-hedging price,
we have $g(t,x,n)\le ng(t,x,1)$. As a first step, we have computed the relative infimum super-hedging prices per unit of claims, i.e. $V_0^{*,n}/n$ at time $0$, for different values of $n$ on each period of the test data set. Then, computing the average of the $V_0^{*,n}/n$ values over all the periods, we get that the empirical average of $V_0^{*,n}/n$ is approximately equal to $49.48\%$ for $n=1,5,10,15,20$.  Nevertheless, we observe that the price function per unit of unit of claims, i.e. $g(0,S_0,n)/n$ is non-increasing when $n$ increases, see Figure \ref{F3}. This implies that the equality $g(t,x,n)= ng(t,x,1)$ does not hold.  In Figure \ref{F3}, we clearly observe the convergence of $x\mapsto g(0,x,n)/n$ as $n\to \infty$.

The same question arises for the optimal strategy associated to $V^{*,n}$, i.e. do we have 
$\theta^*(t,S_t,n)=n \theta^*(t,S_t,1)$? Intuitively, this is a priori not the case as $\theta^*(t,S_t,1)=\theta^*(t,S_t,n)/n$ could be not integer-valued. This is confirmed at time $0$ when we compute the optimal strategy $\theta^*(0,S_0,n)/n$  per unit of claims. This is illustrated by Figure \ref{F1}.2   where we compare $\theta^*(0,S_0,n)/n$ for $n=10$ to $\theta^*(0,S_0,1)$. We may observe that the optimal strategy per unit of claims $\theta^*(0,S_0,n)/n$ (blue graph) is smaller that $\theta^*(0,S_0,1)$ for $n=10$.

\begin{figure}[h!]
\includegraphics[scale=.4]{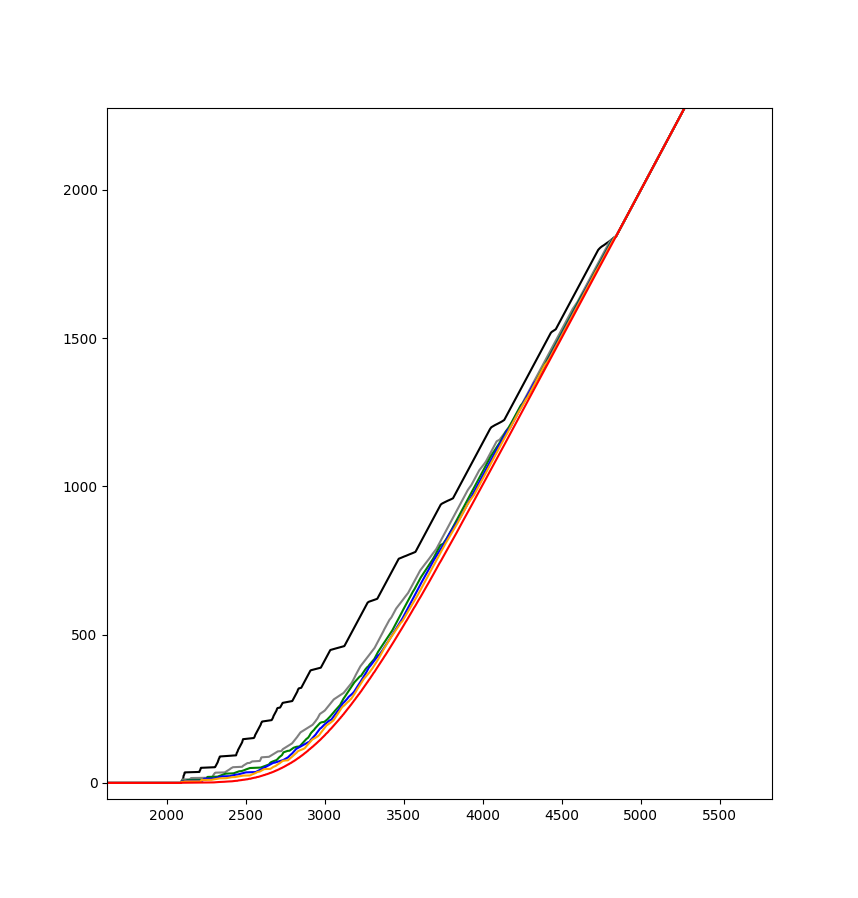} 
\caption{Super-hedging price mapping $x\mapsto g(0,x,n)/n$ of $n$ units of call option per unit of claims for different values of $n=1$ (black), $n=3$ (grey), $n=5$ (green), $n=7$ (blue), $n=10$ (orange), $n=100$ (red).}
\label{F3}
 \end{figure}


\smallskip

































\bigskip




\newpage

\begin{thebibliography}{100}





\bibitem{BCL} Baptiste J., Carassus L. and L\'epinette E. Pricing without martingale measures. Preprint. \url{https://hal.archives-ouvertes.fr/hal-01774150}.

\bibitem{BI} Bienstock D. Computational study of a family of mixed-integer quadratic programming problems. Math. Programming, 74, 121–140, 1996.

\bibitem{BL} Bonami P.  and  Lejeune M. A. An exact solution approach for portfolio optimization problems under stochastic and integer constraints. Oper. Res., 57, 650–670, 2009.

\bibitem{BS} Black F. and Scholes M. The pricing of options and corporate liabilities.
Journal of Political Economy, 81, 3, 637-659, 1973.

\bibitem{BT} Baumann P. and  Trautmann N. Portfolio-optimization models for small investors. Math. Methods Oper. Res., 77, pp. 345–356, 2013.


\bibitem{CL} Carassus L. and L\'epinette E.  Pricing without no-arbitrage condition in discrete-time.  Journal of Mathematical Analysis and Applications, 505, 1, 125441, 2022. 

\bibitem{DLW}  Deng X.,  Li Z. and  Wang S. Computational complexity of arbitrage in frictional security market. International Journal of Foundations of Computer Science, 13, 681–684, 2002.

\bibitem{DMW} Dalang E.C., Morton A. and Willinger W. Equivalent martingale measures
and no-arbitrage in stochastic securities market models. Stochastics
and Stochastic Reports, 29, 185-201, 1990.

\bibitem{DS1} Delbaen F. and Schachermayer W. A general version of the fundamental
theorem of asset pricing. Mathematische Annalen, 300, 463-520, 1994.

\bibitem{DS2} Delbaen F. and Schachermayer W. The fundamental theorem of asset
pricing for unbounded stochastic processes. Mathematische Annalen,
312, 215-250, 1996.

\bibitem{DS3} Delbaen F. and Schachermayer W. The mathematics of arbitrage.
Springer Finance, Ed. 2006, 2nd printing, XVI, 2008.


\bibitem{GK} Gerhold S. and Kr\"uhner P. Dynamic trading under trading constraints. Finance and Stochastics, 22, 919-957, 2018.



\bibitem{EL} El Mansour Meriem and L\'epinette E.  Conditional interior and conditional closure of a random sets). Journal of Optimization Theory and Applications, 5051 187, 356-369, 2020.
\end{thebibliography}
\end{document}